\documentclass{jac}

\usepackage[english
]{babel}
\usepackage[utf8]{inputenc}
\usepackage{amssymb,stmaryrd,pifont,amsmath,tikz,subfig}

\begin{document}

\title[Infinite Time Cellular Automata]{Infinite Time Cellular Automata:\\ a Real Computation Model}
\author
{F. Givors}{Fabien Givors}
\author
{G. Lafitte}{Gregory Lafitte}
\author
{N. Ollinger}{Nicolas Ollinger}
\address
{Laboratoire d'Informatique Fondamentale de Marseille (LIF) \\ CNRS -- Aix-Marseille Université\\39 rue Joliot-Curie\\13453
  Marseille Cedex 13, France} 

\thanks{The research presented in this paper has been made possible by the support of the French ANR grants {\em NAFIT} (ANR-08-DEFIS-008-01) and {\em EMC} (ANR-09-BLAN-0164-01).}

\begin{abstract}\noindent
We define a new transfinite time model of computation, {\em infinite time cellular automata}. The model is shown to be as powerful than infinite time Turing machines, both on finite and infinite inputs; thus inheriting many of its properties. We then show how to simulate the canonical real computation model, {\em BSS machines}, with infinite time cellular automata in exactly $\omega$ steps.
\end{abstract}

\maketitle

\section*{Introduction}

When the second and third authors of this paper were in their PhD years, their common advisor, Jacques Mazoyer, had encouraged them to define a computation model on real numbers using cellular automata. The second author had defined at the end of the Nineties a cellular automata generalization running into transfinite time, but it was never published and was only clumsily defined in his PhD thesis. They thought of using it as a real computation model in 2001 in Riga but never did anything about it since. This paper is a description of these ideas which had stayed for ten years in the state of scribbled notes.

Transfinite time computation models were first considered in 1989 by Hamkins and Kidder. Hamkins and Lewis later developed the model of infinite time Turing machines and its theory in \cite{HL00}. Koepke \cite{Koe06} defined another transfinite time model based on register machines, that was later refined by Koepke and Miller \cite{KM08}. These models differ in their computation power, the infinite time Turing machines model being the most powerful one.

In \cite{BSS89}, Blum, Shub and Smale introduced, also in 1989, a model of computation, coined {\em BSS machines}, intended to describe computations over the real numbers. It can be viewed as Turing machines with tapes whose cells (or Random Access Machines with registers that) can store arbitrary real numbers and that can compute rational functions overs reals at unit cost. Their model, which is more general\footnote{It really provides a setting for computing over rather arbitrary structures.} than the presentation provided in this paper, is a canonical model of real computation. The idea of real computation is to deal with hypothetical computing machines using infinite-precision real numbers and to be able to prove results of computations operating on the set of real numbers. A typical example is studying the computability of the Mandelbrot set. It can be viewed as an idealized analog computer which operates on real numbers and is differential, whereas digital computers are limited to integers and are algebraic. For a survey of analog computations, the reader is referred to the survey \cite{Bou08}.

In this paper, we introduce a transfinite time computation model, the {\em infinite time cellular automata}.
The model is arguably more {\em natural} and {\em uniform} than other transfinite time models introduced, for the same reasons cellular automata are more {\em natural} and {\em uniform} than Turing machines. There is no head wandering here and there and there is no difference between states and data.

We show that the infinite time cellular automata have the same computing power than infinite time Turing machines, both on finite and infinite inputs. They thus inherit the nice properties of this latter model. We then show how to simulate the BSS machines with infinite time cellular automata in exactly $\omega$ steps.
We finish by introducing another transfinite time model based on cellular automata.

\section{Infinite time cellular automata}
\label{ifac}

\begin{defi}
\label{def:itac}

An {\em infinite time cellular automaton} (ITCA) $A$ is defined by $\Sigma$, the finite
set of states of $A$, linearly ordered by $\prec$ and with a least element
$\mathbf{0}$;
and $\delta : \Sigma^3 \to \Sigma$, the local rule of $A$,
satisfying $\delta(\mathbf{0},\mathbf{0},\mathbf{0})=\mathbf{0}$, so that
$\mathbf{0}$ is a {\em quiescent} state.

A {\em configuration} is an element of $\Sigma^{\mathbb{Z}}$.
The local rule $\delta$ induces a global rule $\Delta: \Sigma^{\mathbb{Z}} \to \Sigma^{\mathbb{Z}}$ on configurations such that
$\Delta(C)_i = \delta (C_{i-1},C_i,C_{i+1})$ for $i\in\mathbb{Z}$ and $C\in\Sigma^{\mathbb{Z}}$.

Starting from a configuration $C\in\Sigma^{\mathbb{Z}}$, the {\em evolution} of length $\theta\in\text{Ord}$ of $A$ is given by $(\Delta^\alpha(C))_{\alpha \leqslant \theta}$:
$$
\begin{array}{lcl}
\Delta^{\beta+1}(C) & = &\Delta (\Delta^{\beta} (C)) \\
\Delta^{\lambda}(C)_i & =  &\liminf^{\prec}_{\gamma<\lambda} \Delta^{\gamma}(C)_i \text{ for all $i\in \mathbb{Z}$ and $\lambda$ limit} 
\end{array}
$$
\end{defi}

\begin{defi}
Let $\mathbf{h}\in\Sigma$ a particular state we will refer to as the {\em halting} state.

An evolution of length $\theta$ is called a {\em computation} if the state $\mathbf{h}$ appears in the last configuration, $\Delta^\theta(C)$, but not before this stage. 
\end{defi}

We settle a convention on the way we code integers or real numbers in our model. For example, we could code integers and real numbers in binary (using $\mathbf{0}$ and another state as symbols for $0$ and $1$) on the right cells (cells whose indices belong to $\mathbb{N}$).

\begin{defi}
Let $X$ be a space for which a coding has been settled. (For example, $\mathbb{N}$, $\mathbb{R}$, $2^\mathbb{N}$ or $2^\mathbb{Z}$.)

A (partial) function $F$ on $X$, $F: X \to X$, is said to be {\em infinite time computable} if
there is an infinite time cellular automaton such that for each $x\in\mathop{dom}(F)$, there is
a computation starting with a configuration with a coding of $x$ that halts on a configuration with a similar coding of $F(x)$.

A set $A\subseteq X$ is {\em infinite time decidable} if its characteristic function is infinite time computable, and is {\em infinite time semi-decidable} if it is the domain of an infinite time computable function.
\end{defi}

We use the term ``semi-decidable" instead of ``enumerable", since contrarily to the classical computability concepts, in the transfinite time context, being semi-decidable is not equivalent to being the range of a computable function.

\section{Properties of infinite time cellular automata}

\subsection{Comparisons with other infinite time models}

Hamkins, Kidder and Lewis \cite{HL00} have defined an infinite time Turing machines model and Koepke \cite{Koe06} has defined an infinite time register machines model, that was later refined by Koepke and Miller \cite{KM08}.

The {\em infinite time Turing machines} (ITTM) work as a classical Turing machine with a head reading and writing on a bi-infinite tape, moving left and right in accordance with the instructions of a finite program with finitely many states. At successor stages of computation, the machine operates in exactly the classical manner. At limit stages, the machine enters a special limit state, with the head on the origin cell, and each cell of the tape taking the value of the lim inf of the values appearing in that cell before that limit stage.

The {\em infinite time register machines} behave like standard register machines at successor stages.
At limit times, the register contents are defined using lim inf's of the previous register contents.
The difficulty here is that the lim inf does not necessarily exist in that case, since a register can contain arbitrary large integers. The machines of \cite{Koe06} crashed in such a case. Those of \cite{KM08} continue beyond such crashes by resetting a register to $0$ whenever it overflows.

The infinite time Turing machines are strictly stronger than infinite time register machines: the halting problem for infinite time register machines can be decided by an ITTM.

\begin{thm}\label{thm:itcapower}
Infinite time cellular automata have the same computing power of infinite time Turing machines.
\end{thm}

\begin{proof}
The right to left implication goes as follows:

At successor stages, the simulation of ITTM by ITCA works the same way it does in the
non-infinite-time case.

At limit stages,
we have to put the configuration of the ITCA in the limit configuration
simulation of the ITTM simulated.
To do this, we need to be able to know that we are at a limit stage.
It suffices to have two adjacent cells of the ITCA at the origin that, at
successor stages, alternate between $\mathbf{0}$ and some other state such that
the adjacent states are different. At the next limit stage, they will be both equal to
$\mathbf{0}$. We also have to use the same trick on the cells visited by the head to make sure that at limit stages, if the ITTM becomes stationary, we can wipe out the stationary state from our simulation tape and enter in the special limit state. The ITCA can then prepare the configuration to continue the ITTM
simulation.

The left to right implication: it takes $\omega$ steps with an ITTM machine to
simulate an ITCA global step (on an infinite input). It is just then a matter of
determining whether the ITTM is at a limit stage or not. This is easily achieved by
an ITTM since it enters a special limit state at limit stages. 
\end{proof}

\subsection{Features of those infinite time models}

Hamkins, Kidder, Lewis and Welch \cite{HL00,Wel99,Wel00a,Wel00b} have shown many properties of infinite time Turing machines. By Theorem \ref{thm:itcapower}, infinite time cellular automata have many of these same properties. We state in the following the properties inherited by infinite time cellular automata.

\begin{thm}
The set of reals coding well-orders is infinite time decidable.
	
The hyperarithmetic sets are those that are decidable in time less than some recursive ordinal. Every $\Pi^1_1$ set is decidable and the class of decidable sets is contained in $\Delta^1_2$.
\end{thm}

\begin{defi}
An ordinal $\alpha$ is {\em clockable} if there is an ITCA computation starting from the all-but-one quiescent configuration $C$ ($C_0\neq \mathbf{0}$ and $C_i = \mathbf{0} \ \forall i\in \mathbb{Z}\setminus\{0\} $) and that halts after exactly $\alpha$ steps (meaning that the $\alpha^{\text{th}}$ configuration, $\Delta^\alpha(C)$, is the first configuration in which the halting state $\mathbf{h}$ appears).

A real $r$ is {\em writable} if it is the output of an ITCA computation. An ordinal is writable if it is coded by such a real.
\end{defi}

There are of course only countably many clockable and writable ordinals, since there are only countably many local rules.

\begin{thm}
Every recursive ordinal is clockable. Even $\omega^{\text{CK}}_1 + \omega$ is clockable.  Beyond that, there are many intervals of non-clockable ordinals. The supremum of clockable ordinals is recursively inaccessible\footnote{A {\em recursively inaccessible} ordinal is an ordinal that is both admissible and a limit of admissibles. An ordinal $\alpha$ is {\em admissible} if the construction of the Gödel universe, $L$, up to a stage $\alpha$, yields a model $L_\alpha$ of Kripke-Platek set theory. The Church-Kleene ordinal, $\omega^{\text{CK}}_1$, is the smallest non-recursive ordinal and is the smallest admissible ordinal. For more on admissibles, see the most excellent book of Barwise \cite{Bar75}.}. Moreover, the writable ordinals however form an initial segment of the ordinals. The supremum of the writable ordinals is the supremum of the clockable ordinals.
\end{thm}

One of the beautiful theorems of ITTMs that carry through to ITCAs is the Lost Melody Theorem.
The real constructed in this theorem is like a lost melody that you can recognize when someones hums it to you, but which you cannot sing on your own. 

\begin{thm}[Lost Melody Theorem]
There is a real $r$ which is recognizable ($\{r\}$ is decidable), but not writable.
\end{thm}

There are different ways to construct such lost melody reals. 
One way is to consider the supremum $\gamma$ of the ordinal stages by which an ITTM computation, on an empty input, either halts or repeats. 
Notice that by a simple cofinality argument, we can show that all these ordinals are countable.
There is a smallest ordinal $\delta \geqslant \gamma$ such that $L_{\delta+1} \models \text{``$\delta$ is countable''}$. $L_{\delta +1}$ has a canonical well-ordering, thus there is some real $r \in L_{\delta+1}$ which is least with respect to the canonical $L$ order, such that $r$ codes $\delta$. 
$\gamma$ (and thus $r$) are somehow a generalization of the busy beaver problem to transfinite time computations. It is then not surprising that $r$ cannot be computable, since that would render the infinite time halting problem decidable. It is recognizable because it is possible to reconstruct the $L$ hierarchy using an ITTM and verify that $r$ is really the least coding of an ordinal having the properties of $\delta$.

\section{Computations on the reals}

\subsection{Blum-Shub-Smale model}

Blum, Shub and Smale \cite{BSS89} introduced the {\em BSS model}. 

A simplified presentation of the BSS model goes through defining ``Turing machines with real numbers''. We follow Hainry's presentation in his PhD thesis \cite{HaiPhD}.

\begin{defi}
A {\em simplified BSS machine} (or shortly, BSS machine) is composed of an infinite tape and a program.
The infinite tape is made of cells, each containing a real number.  
We denote the tape by $(x_n)_{n\in\mathbb{Z}} \in \mathbb{R}^\mathbb{Z}$.
The program is a numbered (finite) sequence of instructions. The number of each instruction in the program sequence is seen as a state ($\in Q$). 
The instructions are
\begin{itemize}
\item {\em go right}: changes the tape to $(x_{n+1})_{n\in\mathbb{Z}}$;
\item {\em go left}: changes the tape to $(x_{n-1})_{n\in\mathbb{Z}}$;
\item {\em branch if greater than $0$}: if the current cell ($x_0$) is greater than $0$, then it branches to a specified location in the program;
\item {\em branch if equal to $0$}: if the current cell ($x_0$) is equal to $0$, then it branches to a specified location in the program;
\item {\em make a computation}: the current cell ($x_0$) is changed to be equal to the result of a computation from $x_0$, $x_1$ and possible constants ($k\in\mathbb{R}$). A computation is one of the following:
\begin{itemize}
\item $x_0 \shortleftarrow -x_0$;
\item $x_0 \shortleftarrow k$;
\item $x_0 \shortleftarrow x_0+x_1$;
\item $x_0 \shortleftarrow x_0\times x_1$;
\item $x_0 \shortleftarrow k\times x_0$.
\end{itemize}
\end{itemize}
The machine starts by executing the first instruction (with the least number) of the program and continues by executing the next instruction and so on until it branches. It halts when it has no more instructions to execute.
\end{defi}

It is possible to give a definition for a more general BSS machine on rather arbitrary structures. In this paper, we will stick with simplified BSS machines on $\mathbb{R}$.

For a lot more on the {\em BSS model} and computation on the real numbers, the reader is referred to the book by Blum, Cucker, Shub and Smale \cite{BCSS98}.

\subsection{BSS by $\omega$-ITCA}

We show how to simulate a simplified BSS machine with an ITCA.

\begin{thm}
A simplified BSS machine can be simulated by an ITCA in $\omega$ steps.
\end{thm}

\begin{proof}
Consider a BSS machine. At each time step $t$ of a computation, only a finite number of non-zero real numbers are defined: $k$ constants inside the program and $l\leqslant t$ cells of the tape. For the sake of clarity, suppose that the BSS machine works on reals in the interval $[-1,1]$\footnote{The construction extends to $\mathbb{R}$ at the cost of non significant tricks, for example by adding just after the bit of sign the encoding of the integer part of each real on a same number of bits followed by a dot.}. Imagine that we encode these $k+l$ reals as infinite words on the alphabet $\left\{0,1\right\}$ encoding a bit of sign followed by their binary expansion.

Each computation instruction of the BSS machine has the nice property that it can be computed, in time $\omega$, by a Turing machine working synchronously on the representation of its operands. For each finite initial portion of the tape, it is left untouched by such a machine after some finite time. Moreover, at the cost of some extra bookkeeping on the tape, such a computation can be achieved in a reversible way\footnote{Just store the non injective choices on a stack that the head pushes in front of itself.}. 

Each branch instruction of the BSS machine can be achieved in a similar way: one can choose a initial hypothesis on the branch\footnote{For branch if equal to zero the hypothesis is that $x_0=0$.} and start a computation by a Turing machine working synchronously on the representation of the operands ; either the machine eventually halts contradicting the hypothesis, or its head moves infinitely towards the end of the infinite words. For each finite initial portion of the tape, it is left untouched by such a machine after some finite time and the computation can be achieved in a reversible way. 

Packing it all together, we can simulate a BSS machine if we can launch as many Turing computation threads as needed. Encode the $k+l$ reals encodings as a single infinite word and put some finite control at the beginning of it. The finite control plays the role of the head of the BSS machine, keeping the current state (instruction number). The control is responsible for launching the next instruction: each time there is enough room on its right, it starts a new thread for the current instruction. If the instruction is a move instruction, the thread simply selects the next cell, adding a new $0$ real to the tuple if necessary. If the instruction is a branch instruction, the control takes the default hypothesis on the result and launches the branch thread described before. If the instruction is a computation instruction, the control launches the thread described above. At each time step, the control is pointing to a current instruction and a finite number of threads are computing on its right. Each thread works on a finite portion of tape where it should have exclusive access. Somewhere on its right is the last point where he modified the reals. When a thread wants to access a portion of the reals already modified by the next thread on its left, it enforce this thread to undo its computation to restore the reals as they were before. A cascade of undoing occurs in such a situation, each thread undoing the next thread. At some point an undone thread reaches the control and forces the control part to start again from the instruction where this thread was started, removing the thread from the computing area. As each thread eventually goes to infinity, such an inefficient compute/uncompute dance converges. The same applies when a branch thread discovers that its hypothesis was wrong: it goes back to control to take the other branch of computation, forcing threads newer than him to undo. Notice that a key point of the construction is that there is always enough room for bookkeeping: if a thread needs more space, it just can wait for more space to be available before continuing its computation, either it will eventually happen, or a backtrack will occur that can be handled.

The result of a computation is obtained easily: at time $\omega$, if the control eventually converged to an accepting state, the control part encodes an accepting state and the encoded reals contain the result of the computation ; if the BSS machine did not converge, the control part does not encode an accepting state.

Let us now explain how the described simulation can be carried on a ITCA. Given a BSS machine, the ITCA is constructed as follows. Only a semi-infinite part of the configuration is used. Cell $\#0$ encodes control and the cells on its right encode the reals, the computation area. The computation area is constructed in 3 layers, as depicted on figure~\ref{fig:layers}. The data layer encodes the bits of real numbers and the current position of the head: it is divided into blocks of length $k+l$, separated by $\#$ border symbols, containing a bit for each real, the current position of the head being circled. The stack layer encodes the working area for the threads, where each head has its computing stack, and the boundaries of the threads: each thread delimits monotonically the area it already modified by a $*$ symbol. The head layer encodes the heads of the threads, the active parts of the ITCA.

Each instruction of the BSS machine is simulated as explained before. A move thread simply moves the circle to the previous or next bit, adding a new real if necessary, using its thread stack. A branch thread does not modify reals but checks if the hypothesis was true or false, as depicted on figure~\ref{fig:branch}. A computation thread modifies the reals, making choices (for example the values of carries when adding two reals), backtracking when the choice was a bad one, as depicted for addition on figure~\ref{sfig:add}. Notice that multiplication can be significantly simplified by the fact that we can create new threads, thus it can be decomposed into additions launched by a master thread, as depicted on figure~\ref{sfig:mul}.

It is important to notice that the monotonicity of the forward movement of $*$ symbols ensures that there is no concurrency problem due to neighbor threads backtracking and coming back forward in a same area: when a thread wants to access an area already explored by its follower, the follower is asked to undo its computation. To ask a neighbor to undo, a thread simply modifies its neighbor $*$ symbol to inform him.

The details of the construction use rather classical but tedious CA encoding tricks. The key argument of the proof is that an ITCA can simulate in time $\omega$ the work of an unbounded number of Turing heads, thus achieving the same quantity of work than a ITTM in time $\omega^2$.
\end{proof}

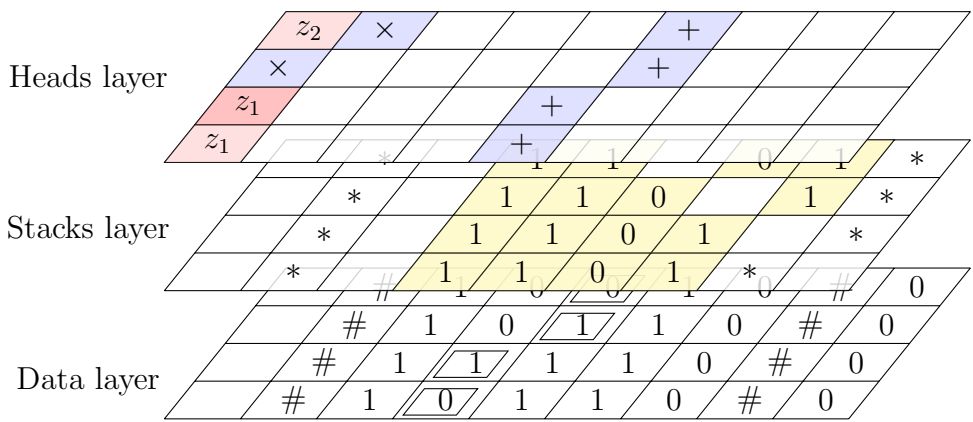
\begin{figure}
\centering
\begin{tikzpicture}
  \begin{scope}
  \pgftransformcm{1}{0}{0.4}{0.5}{\pgfpoint{0cm}{0cm}}
  \fill[white,opacity=0.8] (-1,0) rectangle (8,4);
  \draw [black,step=1cm] (-1,0) grid (8,4);
  \draw (0.5, 0.5) node {$\#$};
  \draw (1.5, 0.5) node {$1$};
  \draw (2.5, 0.5) node {$0$};
  \draw (2.1, 0.1) rectangle (2.8, 0.8);
  \draw (3.5, 0.5) node {$1$};
  \draw (4.5, 0.5) node {$1$};
  \draw (5.5, 0.5) node {$0$};
  \draw (6.5, 0.5) node {$\#$};
  \draw (7.5, 0.5) node {$0$};
  \draw (0.5, 1.5) node {$\#$};
  \draw (1.5, 1.5) node {$1$};
  \draw (2.5, 1.5) node {$1$};
  \draw (2.1, 1.1) rectangle (2.8, 1.8);
  \draw (3.5, 1.5) node {$1$};
  \draw (4.5, 1.5) node {$1$};
  \draw (5.5, 1.5) node {$0$};
  \draw (6.5, 1.5) node {$\#$};
  \draw (7.5, 1.5) node {$0$};
  \draw (0.5, 2.5) node {$\#$};
  \draw (1.5, 2.5) node {$1$};
  \draw (2.5, 2.5) node {$0$};
  \draw (3.5, 2.5) node {$1$};
  \draw (3.1, 2.1) rectangle (3.8, 2.8);
  \draw (4.5, 2.5) node {$1$};
  \draw (5.5, 2.5) node {$0$};
  \draw (6.5, 2.5) node {$\#$};
  \draw (7.5, 2.5) node {$0$};
  \draw (0.5, 3.5) node {$\#$};
  \draw (1.5, 3.5) node {$1$};
  \draw (2.5, 3.5) node {$0$};
  \draw (3.5, 3.5) node {$0$};
  \draw (3.1, 3.1) rectangle (3.8, 3.8);
  \draw (4.5, 3.5) node {$1$};
  \draw (5.5, 3.5) node {$0$};
  \draw (6.5, 3.5) node {$\#$};
  \draw (7.5, 3.5) node {$0$};
  \end{scope}

  \begin{scope}
  \pgftransformcm{1}{0}{0.4}{0.5}{\pgfpoint{0cm}{1.7cm}}
  \fill[white,opacity=0.8] (-1,0) rectangle (8,4);
  \draw [black,step=1cm] (-1,0) grid (8,4);
  \draw (0.5, 0.5) node {$*$};
  \draw (1.5, 0.5) node {};
  \draw[fill=yellow!40, opacity=0.6] (2,0) rectangle (3,1);
  \draw (2.5, 0.5) node {$1$};
  \draw[fill=yellow!40, opacity=0.6] (3,0) rectangle (4,1);
  \draw (3.5, 0.5) node {$1$};
  \draw[fill=yellow!40, opacity=0.6] (4,0) rectangle (5,1);
  \draw (4.5, 0.5) node {$0$};
  \draw[fill=yellow!40, opacity=0.6] (5,0) rectangle (6,1);
  \draw (5.5, 0.5) node {$1$};
  \draw (6.5, 0.5) node {$*$};
  \draw (7.5, 0.5) node {$ $};

  \draw (0.5, 1.5) node {$*$};
  \draw (1.5, 1.5) node {};
  \draw[fill=yellow!40, opacity=0.6] (2,1) rectangle (3,2);
  \draw (2.5, 1.5) node {$1$};
  \draw[fill=yellow!40, opacity=0.6] (3,1) rectangle (4,2);
  \draw (3.5, 1.5) node {$1$};
  \draw[fill=yellow!40, opacity=0.6] (4,1) rectangle (5,2);
  \draw (4.5, 1.5) node {$0$};
  \draw[fill=yellow!40, opacity=0.6] (5,1) rectangle (6,2);
  \draw (5.5, 1.5) node {$1$};
  \draw (6.5, 1.5) node {$ $};
  \draw (7.5, 1.5) node {$*$};

  \draw (0.5, 2.5) node {$*$};
  \draw (1.5, 2.5) node {};
  \draw[fill=yellow!40, opacity=0.6] (2,2) rectangle (3,3);
  \draw (2.5, 2.5) node {$1$};
  \draw[fill=yellow!40, opacity=0.6] (3,2) rectangle (4,3);
  \draw (3.5, 2.5) node {$1$};
  \draw[fill=yellow!40, opacity=0.6] (4,2) rectangle (5,3);
  \draw (4.5, 2.5) node {$0$};
  \draw (5.5, 2.5) node {$ $};
  \draw[fill=yellow!40, opacity=0.6] (6,2) rectangle (7,3);
  \draw (6.5, 2.5) node {$1$};
  \draw (7.5, 2.5) node {$*$};

  \draw (0.5, 3.5) node {$*$};
  \draw (1.5, 3.5) node {};
  \draw[fill=yellow!40, opacity=0.6] (2,3) rectangle (3,4);
  \draw (2.5, 3.5) node {$1$};
  \draw[fill=yellow!40, opacity=0.6] (3,3) rectangle (4,4);
  \draw (3.5, 3.5) node {$1$};
  \draw (4.5, 3.5) node {$ $};
  \draw[fill=yellow!40, opacity=0.6] (5,3) rectangle (6,4);
  \draw (5.5, 3.5) node {$0$};
  \draw[fill=yellow!40, opacity=0.6] (6,3) rectangle (7,4);
  \draw (6.5, 3.5) node {$1$};
  \draw (7.5, 3.5) node {$*$};
  \end{scope}

  \begin{scope}
  \pgftransformcm{1}{0}{0.4}{0.5}{\pgfpoint{0cm}{3.4cm}}
  \fill[white,opacity=0.8] (-1,0) rectangle (8,4);
  \draw [black,step=1cm] (-1,0) grid (8,4);
  \draw[fill=red!20, opacity=0.6] (-1,0) rectangle (0,1);
  \draw (-0.5,0.5) node {$z_1$};
  \draw[fill=red!40, opacity=0.6] (-1,1) rectangle (0,2);
  \draw (-0.5,1.5) node {$z_1$};
  \draw[fill=blue!20, opacity=0.6] (-1,2) rectangle (0,3);
  \draw (-0.5,2.5) node {$\times$};
  \draw[fill=red!20, opacity=0.6] (-1,3) rectangle (0,4);
  \draw (-0.5,3.5) node {$z_2$};

  \draw[fill=blue!20, opacity=0.6] (0,3) rectangle (1,4);
  \draw (0.5,3.5) node {$\times$};

  \draw[fill=blue!20, opacity=0.6] (3,0) rectangle (4,1);
  \draw (3.5,0.5) node {$+$};
  \draw[fill=blue!20, opacity=0.6] (3,1) rectangle (4,2);
  \draw (3.5,1.5) node {$+$};
  \draw[fill=blue!20, opacity=0.6] (4,2) rectangle (5,3);
  \draw (4.5,2.5) node {$+$};
  \draw[fill=blue!20, opacity=0.6] (4,3) rectangle (5,4);
  \draw (4.5,3.5) node {$+$};
  \end{scope}

\draw (-2,4.5) node {Heads layer};
\draw (-2,2.5) node {Stacks layer};
\draw (-2,0.5) node {Data layer};

\end{tikzpicture}
\caption{Layers encoding computation}\label{fig:layers}
\end{figure}

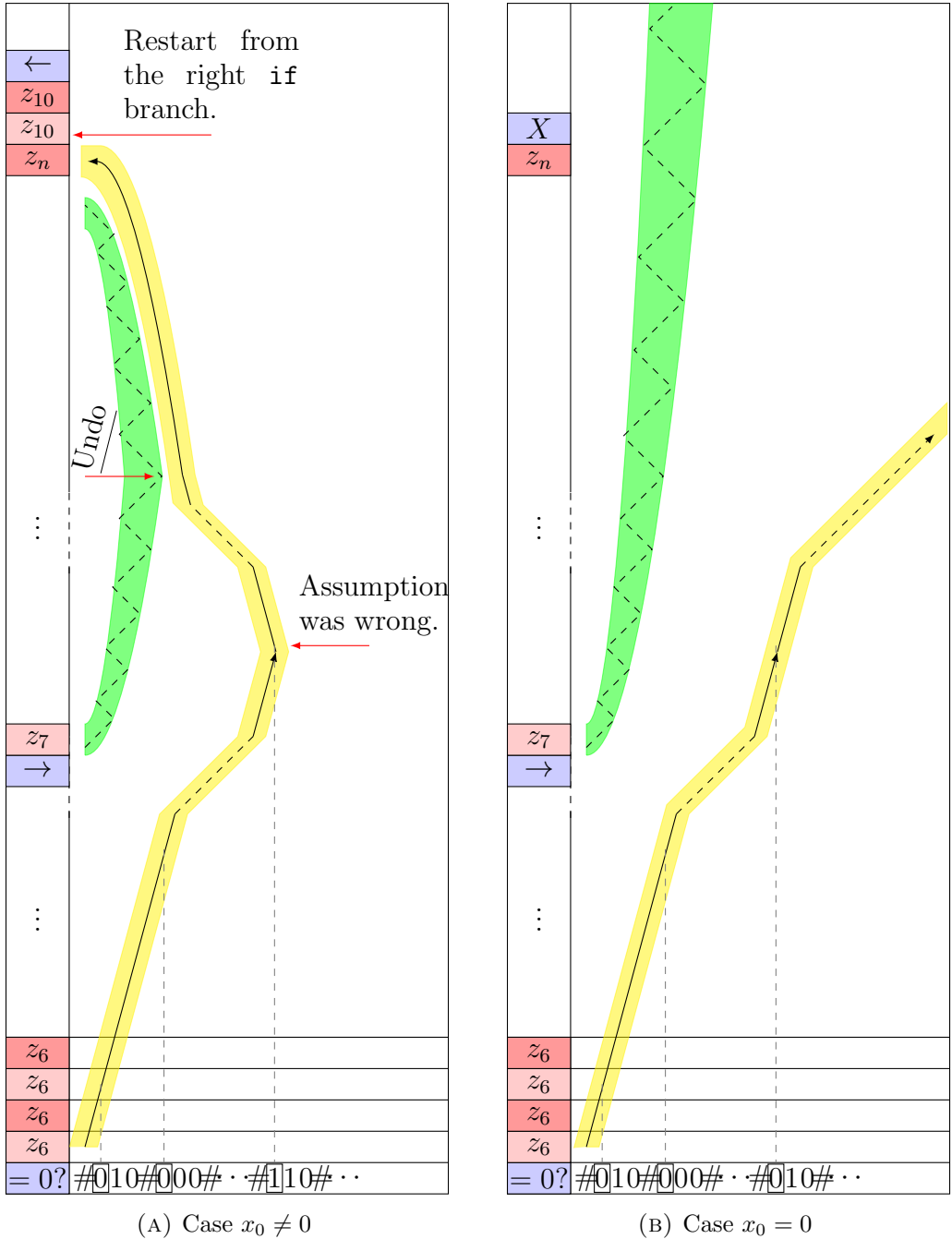
\begin{figure}
\centering
\subfloat[][Case $x_0\neq 0$]{%
\begin{tikzpicture}[scale=0.9]
  \newlength{\rh}
  \setlength{\rh}{19cm}
  \newlength{\rw}
  \setlength{\rw}{7cm}
  \newlength{\lh}
  \setlength{\lh}{0.5cm}
  \newlength{\cw}
  \setlength{\cw}{0.5cm}
  \draw [black,fill=white] (0,0) rectangle (\rw,\rh);
  \draw  (2\cw,0) -- (2\cw,6.1);
  \draw [dashed] (2\cw,6.0) -- (2\cw,7.3);
  \draw  (2\cw,7.3) -- (2\cw,10);
  \draw [dashed] (2\cw,9.9) -- (2\cw,11.2);
  \draw  (2\cw,11.2) -- (2\cw,\rh);

  \foreach \i in {1, 2, 3, 4, 5} {
    \draw [black] (0,\lh*\i) -- (\rw,\lh*\i);
  }

  \draw [black] (2\cw+01*0.5\cw, 0.5\lh) node {$\#$};
  \draw [black] (2\cw+02*0.5\cw, 0.5\lh) node {$0$};
  \draw (2.75\cw, 0.05\lh) rectangle (3.25\cw, 0.95\cw);
  \draw [black] (2\cw+03*0.5\cw, 0.5\lh) node {$1$};
  \draw [black] (2\cw+04*0.5\cw, 0.5\lh) node {$0$};
  \draw [black] (2\cw+05*0.5\cw, 0.5\lh) node {$\#$};
  \draw [black] (2\cw+06*0.5\cw, 0.5\lh) node {$0$};
  \draw (4.75\cw, 0.05\lh) rectangle (5.25\cw, 0.95\cw);
  \draw [black] (2\cw+07*0.5\cw, 0.5\lh) node {$0$};
  \draw [black] (2\cw+08*0.5\cw, 0.5\lh) node {$0$};
  \draw [black] (2\cw+09*0.5\cw, 0.5\lh) node {$\#$};
  \draw [black] (2\cw+10.5*0.5\cw, 0.5\lh) node {$\cdots$};
  \draw [black] (2\cw+12*0.5\cw, 0.5\lh) node {$\#$};
  \draw [black] (2\cw+13*0.5\cw, 0.5\lh) node {$1$};
  \draw (8.25\cw, 0.05\lh) rectangle (8.75\cw, 0.95\cw);
  \draw [black] (2\cw+14*0.5\cw, 0.5\lh) node {$1$};
  \draw [black] (2\cw+15*0.5\cw, 0.5\lh) node {$0$};
  \draw [black] (2\cw+16*0.5\cw, 0.5\lh) node {$\#$};
  \draw [black] (2\cw+17.5*0.5\cw, 0.5\lh) node {$\cdots$};

  \draw [fill=blue!20] (0,0\lh) rectangle (2\cw,1\lh);
  \draw (\cw, 0.5\lh) node {$=0?$};
  \draw [fill=red!20] (0,1\lh) rectangle (2\cw,2\lh);
  \draw (\cw, 1.5\lh) node {$z_6$};
  \draw [fill=red!40] (0,2\lh) rectangle (2\cw,3\lh);
  \draw (\cw, 2.5\lh) node {$z_6$};
  \draw [fill=red!20] (0,3\lh) rectangle (2\cw,4\lh);
  \draw (\cw, 3.5\lh) node {$z_6$};
  \draw [fill=red!40] (0,4\lh) rectangle (2\cw,5\lh);
  \draw (\cw, 4.5\lh) node {$z_6$};

  \draw (\cw, 9\lh) node { \vdots };

  \draw [fill=blue!20] (0,13\lh) rectangle (2\cw,14\lh);
  \draw (\cw, 13.5\lh) node {$\rightarrow$};
  \draw [fill=red!20] (0,14\lh) rectangle (2\cw,15\lh);
  \draw (\cw, 14.5\lh) node {$z_7$};

  \draw (\cw, 21.5\lh) node { \vdots };

  \draw [fill=red!40] (0,32.5\lh) rectangle (2\cw,33.5\lh);
  \draw (\cw, 33\lh) node {$z_n$};
  \draw [fill=red!20] (0,33.5\lh) rectangle (2\cw,34.5\lh);
  \draw (\cw, 34\lh) node {$z_{10}$};
  \draw [fill=red!40] (0,34.5\lh) rectangle (2\cw,35.5\lh);
  \draw (\cw, 35\lh) node {$z_{10}$};
  \draw [fill=blue!20] (0,35.5\lh) rectangle (2\cw,36.5\lh);
  \draw (\cw, 36\lh) node {$\leftarrow$};

  \draw[draw=yellow, fill=yellow, opacity=0.5] (2.5\cw, 1.5\lh) -- ++(0:0.4\lh)
-- ++(75:5.5) -- ++(45:1.75) -- ++(75:1.4) -- ++(105:1.4) -- ++(135:1.4) --
++(105:0.5) parabola[bend at end] ++(-3\cw, 10.5\lh)
-- ++(0:-0.6\lh) -- ++(90:-0.5) parabola ++(2.9\cw, -10.4\lh)
-- ++(135:-1.45) -- ++(105:-1.4) --
++(75:-1.4) -- ++(45:-1.75) -- ++(75:-5.5);

  \draw[draw=green, fill=green, opacity=0.5] (2.5\cw, 14\lh) parabola
++(2.45\cw, 8.9\lh) parabola[bend at end] ++(-2.45\cw, 8.9\lh) -- ++(0,-0.5)
parabola ++(1.25\cw, -7.9\lh) parabola[bend at end] ++(-1.25\cw, -7.9\lh) --
(2.5\cw, 14\lh);
  \draw[dashed] (2.5\cw, 14.25\lh) -- ++(45:0.65) -- ++(135:0.4) -- ++(45:0.7)
-- ++(135:0.5) -- ++(45:0.77) -- ++(135:0.6) -- ++(45:0.85) -- ++(135:0.7) --
++(45:0.95)
-- ++(135:0.95) -- ++(45:0.7) -- ++(135:0.85) -- ++(45:0.6) -- ++(135:0.77) --
++(45:0.5) -- ++(135:0.7) -- ++(45:0.4) -- ++(135:0.65);

  \fill (3\cw, 23\lh) -- node[sloped,above] {Undo} ++(0.5\cw, 2\lh);
  \draw[red, <-, >=latex] (4.7\cw, 22.9\lh) -- (2.5\cw, 22.9\lh);
  \fill (11.5\cw, 17.5\lh) node[sloped,above] {\parbox{2cm}{Assumption was wrong.}} ++(0.5\cw, 2\lh);
  \draw[red, ->, >=latex] (11.5\cw, 17.5\lh) -- (9.0\cw, 17.5\lh);

  \fill (6.5\cw, 34.0\lh) node[sloped,above] {\parbox{2.5cm}{Restart from the right \texttt{if} branch.}} ++(0.5\cw, 2\lh);
  \draw[red, ->, >=latex] (6.5\cw, 33.8\lh) -- (2.1\cw, 33.8\lh);

  \draw (2.5\cw, 1.5\lh) -- ++(75:5.5);
  \draw[dashed] (2.5\cw, 1.5\lh)++(75:5.5) -- ++(45:1.75);
  \draw[->, >=latex] (2.5\cw, 1.5\lh) ++(75:5.5)++(45:1.75) -- ++(75:1.4);
  \draw (2.5\cw, 1.5\lh)++(75:5.5)++(45:1.75)++(75:1.4) -- ++(105:1.4);
  \draw[dashed] (2.5\cw, 1.5\lh)++(75:5.5)++(45:1.75)++(75:1.4)++(105:1.4) -- ++(135:1.4);
  \draw (2.5\cw, 1.5\lh)++(75:5.5)++(45:1.75)++(75:1.4)++(105:1.4)++(135:1.4) -- ++(105:0.5);

  \draw[<-, >=latex] (2.5\cw, 1.5\lh)++(75:5.5)++(45:1.75)++(75:1.4)++(105:1.4)++(135:1.4)++(105:0.5)
  ++(-3\cw, 10\lh) parabola ++(3\cw, -10\lh);

  \draw[dashed,gray] (3\cw, 3.5\lh) -- (3\cw,\lh);
  \draw[dashed,gray] (5\cw, 11\lh) -- (5\cw,\lh);
  \draw[dashed,gray] (8.5\cw, 17.5\lh) -- (8.5\cw,\lh);
\end{tikzpicture}}\hfil
\subfloat[][Case $x_0=0$]{%
\begin{tikzpicture}[scale=0.9]
  \setlength{\rh}{19cm}
  \setlength{\rw}{7cm}
  \setlength{\lh}{0.5cm}
  \setlength{\cw}{0.5cm}
  \draw [black,fill=white] (0,0) rectangle (\rw,\rh);
  \draw  (2\cw,0) -- (2\cw,6.1);
  \draw [dashed] (2\cw,6.0) -- (2\cw,7.3);
  \draw  (2\cw,7.3) -- (2\cw,10);
  \draw [dashed] (2\cw,9.9) -- (2\cw,11.2);
  \draw  (2\cw,11.2) -- (2\cw,\rh);

  \foreach \i in {1, 2, 3, 4, 5} {
    \draw [black] (0,\lh*\i) -- (\rw,\lh*\i);
  }

  \draw [black] (2\cw+01*0.5\cw, 0.5\lh) node {$\#$};
  \draw [black] (2\cw+02*0.5\cw, 0.5\lh) node {$0$};
  \draw (2.75\cw, 0.05\lh) rectangle (3.25\cw, 0.95\cw);
  \draw [black] (2\cw+03*0.5\cw, 0.5\lh) node {$1$};
  \draw [black] (2\cw+04*0.5\cw, 0.5\lh) node {$0$};
  \draw [black] (2\cw+05*0.5\cw, 0.5\lh) node {$\#$};
  \draw [black] (2\cw+06*0.5\cw, 0.5\lh) node {$0$};
  \draw (4.75\cw, 0.05\lh) rectangle (5.25\cw, 0.95\cw);
  \draw [black] (2\cw+07*0.5\cw, 0.5\lh) node {$0$};
  \draw [black] (2\cw+08*0.5\cw, 0.5\lh) node {$0$};
  \draw [black] (2\cw+09*0.5\cw, 0.5\lh) node {$\#$};
  \draw [black] (2\cw+10.5*0.5\cw, 0.5\lh) node {$\cdots$};
  \draw [black] (2\cw+12*0.5\cw, 0.5\lh) node {$\#$};
  \draw [black] (2\cw+13*0.5\cw, 0.5\lh) node {$0$};
  \draw (8.25\cw, 0.05\lh) rectangle (8.75\cw, 0.95\cw);
  \draw [black] (2\cw+14*0.5\cw, 0.5\lh) node {$1$};
  \draw [black] (2\cw+15*0.5\cw, 0.5\lh) node {$0$};
  \draw [black] (2\cw+16*0.5\cw, 0.5\lh) node {$\#$};
  \draw [black] (2\cw+17.5*0.5\cw, 0.5\lh) node {$\cdots$};

  \draw [fill=blue!20] (0,0\lh) rectangle (2\cw,1\lh);
  \draw (\cw, 0.5\lh) node {$=0?$};
  \draw [fill=red!20] (0,1\lh) rectangle (2\cw,2\lh);
  \draw (\cw, 1.5\lh) node {$z_6$};
  \draw [fill=red!40] (0,2\lh) rectangle (2\cw,3\lh);
  \draw (\cw, 2.5\lh) node {$z_6$};
  \draw [fill=red!20] (0,3\lh) rectangle (2\cw,4\lh);
  \draw (\cw, 3.5\lh) node {$z_6$};
  \draw [fill=red!40] (0,4\lh) rectangle (2\cw,5\lh);
  \draw (\cw, 4.5\lh) node {$z_6$};

  \draw (\cw, 9\lh) node { \vdots };

  \draw [fill=blue!20] (0,13\lh) rectangle (2\cw,14\lh);
  \draw (\cw, 13.5\lh) node {$\rightarrow$};
  \draw [fill=red!20] (0,14\lh) rectangle (2\cw,15\lh);
  \draw (\cw, 14.5\lh) node {$z_7$};

  \draw (\cw, 21.5\lh) node { \vdots };

  \draw [fill=red!40] (0,32.5\lh) rectangle (2\cw,33.5\lh);
  \draw (\cw, 33\lh) node {$z_n$};
  \draw [fill=blue!20] (0,33.5\lh) rectangle (2\cw,34.5\lh);
  \draw (\cw, 34\lh) node {$X$};

  \draw[draw=yellow, fill=yellow, opacity=0.5] (2.5\cw, 1.5\lh) -- ++(0:0.4\lh)
-- ++(75:5.5) -- ++(45:1.75) -- ++(75:1.4) -- ++(75:1.4) -- ++(45:3) --
++(90:0.5) -- ++(45:-0.5) -- ++(45:-3) -- ++(75:-1.4) --
++(75:-1.4) -- ++(45:-1.75) -- ++(75:-5.67);

  \draw[draw=green, fill=green, opacity=0.5] (2.5\cw, 14\lh) parabola
++(4\cw, 24\lh) -- ++(0:-1) parabola[bend at end] ++(-2.0\cw, -23.0\lh) --
(2.5\cw, 14\lh);

  \draw[dashed] (2.5\cw, 14.25\lh) -- ++(45:0.65) -- ++(135:0.4) -- ++(45:0.7)
-- ++(135:0.5) -- ++(45:0.77) -- ++(135:0.6) -- ++(45:0.85) -- ++(135:0.7) --
++(45:0.95)
-- ++(135:0.90) -- ++(45:1.1) -- ++(135:0.85) -- ++(45:1.1) -- ++(135:1.0) --
++(45:1.3) -- ++(135:1.2) -- ++(45:1.4) -- ++(135:1.25) -- ++(45:0.5);

  \draw (2.5\cw, 1.5\lh) -- ++(75:5.5);
  \draw[dashed] (2.5\cw, 1.5\lh) ++(75:5.5) -- ++(45:1.75);
  \draw[->, >=latex] (2.5\cw, 1.5\lh) ++(75:5.5)++(45:1.75) -- ++(75:1.4);
  \draw (2.5\cw, 1.5\lh) ++(75:5.5)++(45:1.75)++(75:1.4) -- ++(75:1.4);
  \draw[->, >=latex,dashed] (2.5\cw, 1.5\lh) ++(75:5.5)++(45:1.75)++(75:1.4)++(75:1.4) -- ++(45:3);
  \draw[dashed,gray] (3\cw, 3.5\lh) -- (3\cw,\lh);
  \draw[dashed,gray] (5\cw, 11\lh) -- (5\cw,\lh);
  \draw[dashed,gray] (8.5\cw, 17.5\lh) -- (8.5\cw,\lh);
\end{tikzpicture}}
\caption{\emph{Branch if equal to 0} thread}\label{fig:branch}
\end{figure}

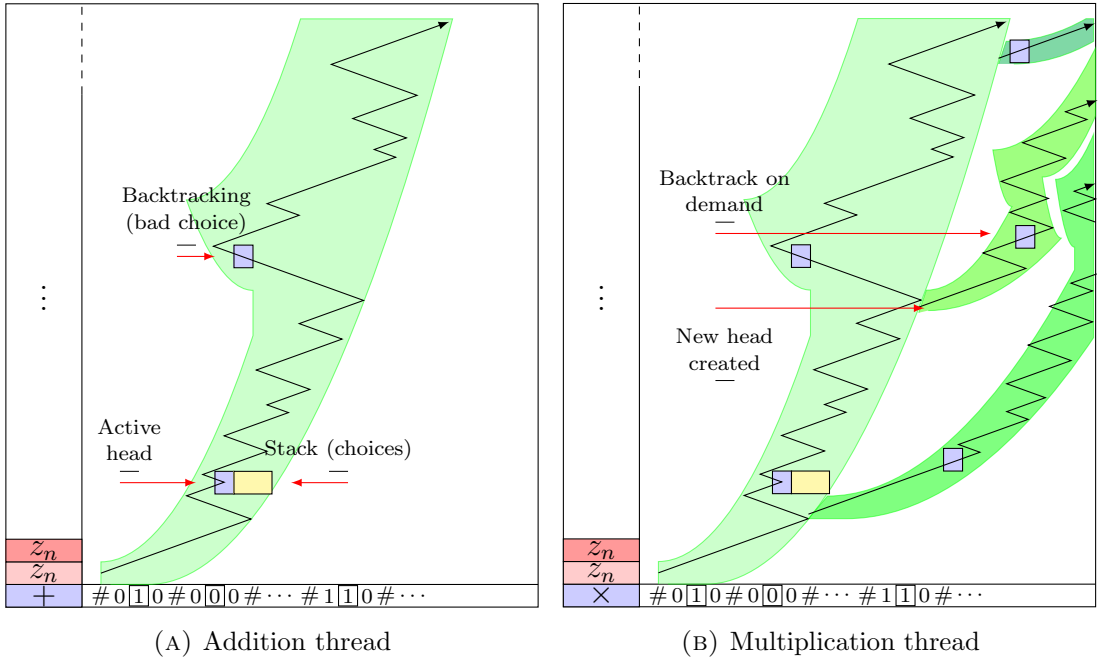
\begin{figure}
\centering
\subfloat[][Addition thread]{\label{sfig:add}%
\begin{tikzpicture}
  \setlength{\rh}{8cm}
  \setlength{\rw}{7cm}
  \setlength{\lh}{0.3cm}
  \setlength{\cw}{0.5cm}
  \draw [black,fill=white] (0,0) rectangle (\rw,\rh);
  \draw  (2\cw,0) -- (2\cw,\rh-4\lh);
  \draw[dashed] (2\cw,\rh-4\lh) -- (2\cw,\rh);

  \draw (0,\lh) -- (\rw,\lh);

  \draw [black] (2\cw+01*0.5\cw, 0.5\lh) node {\tiny $\#$};
  \draw [black] (2\cw+02*0.5\cw, 0.5\lh) node {\tiny $0$};
  \draw [black] (2\cw+03*0.5\cw, 0.5\lh) node {\tiny $1$};
  \draw (3.25\cw, 0.05\lh) rectangle (3.75\cw, 0.95\lh);
  \draw [black] (2\cw+04*0.5\cw, 0.5\lh) node {\tiny $0$};
  \draw [black] (2\cw+05*0.5\cw, 0.5\lh) node {\tiny $\#$};
  \draw [black] (2\cw+06*0.5\cw, 0.5\lh) node {\tiny $0$};
  \draw [black] (2\cw+07*0.5\cw, 0.5\lh) node {\tiny $0$};
  \draw (5.25\cw, 0.05\lh) rectangle (5.75\cw, 0.95\lh);
  \draw [black] (2\cw+08*0.5\cw, 0.5\lh) node {\tiny $0$};
  \draw [black] (2\cw+09*0.5\cw, 0.5\lh) node {\tiny $\#$};
  \draw [black] (2\cw+10.5*0.5\cw, 0.5\lh) node {\tiny $\cdots$};
  \draw [black] (2\cw+12*0.5\cw, 0.5\lh) node {\tiny $\#$};
  \draw [black] (2\cw+13*0.5\cw, 0.5\lh) node {\tiny $1$};
  \draw [black] (2\cw+14*0.5\cw, 0.5\lh) node {\tiny $1$};
  \draw (8.75\cw, 0.05\lh) rectangle (9.25\cw, 0.95\lh);
  \draw [black] (2\cw+15*0.5\cw, 0.5\lh) node {\tiny $0$};
  \draw [black] (2\cw+16*0.5\cw, 0.5\lh) node {\tiny $\#$};
  \draw [black] (2\cw+17.5*0.5\cw, 0.5\lh) node {\tiny $\cdots$};

  \draw [fill=blue!20] (0,0\lh) rectangle (2\cw,1\lh);
  \draw (\cw, 0.5\lh) node {$+$};
  \draw [fill=red!20] (0,1\lh) rectangle (2\cw,2\lh);
  \draw (\cw, 1.5\lh) node {$z_n$};
  \draw [fill=red!40] (0,2\lh) rectangle (2\cw,3\lh);
  \draw (\cw, 2.5\lh) node {$z_n$};

  \draw (\cw, 14\lh) node { \vdots };

  \draw[draw=green, fill=green!40, opacity=0.5] (2.5\cw, 1\lh) -- ++(0,1\lh) parabola
++(4.0\cw, 10.0\lh) -- ++(90:2\lh) parabola ++(-1.75\cw, 4\lh) parabola ++(3\cw,
8\lh) -- ++(0:4\cw) parabola[bend at end] ++(-8\cw, -25\lh) -- ++(0:-1\cw);

  \draw[fill=blue!20] (5.5\cw, 5\lh) rectangle (6\cw, 6\lh);
  \draw[fill=yellow!40] (6\cw, 5\lh) rectangle (7\cw, 6\lh);
  \draw[fill=blue!20] (6.0\cw, 15\lh) rectangle (6.5\cw, 16\lh);

  \draw[->, >=latex] (2.5\cw, 1.5\lh) -- ++(20:7\lh) -- ++(160:3\lh) --
++(20:1.75\lh) -- ++(160:1\lh) -- ++(20:3\lh) -- ++(160:2\lh) -- ++(20:3\lh) --
++(160:1\lh) -- ++(20:2\lh) -- ++(160:2.5\lh) -- ++(20:4\lh) -- ++(160:2\lh) --
++(20:3\lh) -- ++(160:7\lh) -- ++(20:4\lh) -- ++(160:1.5\lh) -- ++(20:6\lh) --
++(160:1\lh) -- ++(20:1.5\lh) -- ++(160:2.5\lh) -- ++(20:3\lh) -- ++(160:4\lh)
-- ++(20:5.5\lh);

  \fill (3\cw, 6\lh) -- node[sloped,above] {\parbox{1cm}{\centering\tiny Active head}} ++(0.5\cw, 0);
  \draw[red, ->, >=latex] (3\cw, 5.5\lh) -- (5.0\cw, 5.5\lh);
  \fill (8.5\cw, 6.0\lh) -- node[sloped,above] {\parbox{2cm}{\centering\tiny Stack (choices)}} ++(0.5\cw, 0);
  \draw[red, ->, >=latex] (9\cw, 5.5\lh) -- (7.5\cw, 5.5\lh);
  \fill (4.5\cw, 16.0\lh) -- node[sloped,above] {\parbox{2cm}{\centering\tiny Backtracking (bad choice)}} ++(0.5\cw, 0);
  \draw[red, ->, >=latex] (4.5\cw, 15.5\lh) -- (5.5\cw, 15.5\lh);

\end{tikzpicture}}\hfil
\subfloat[][Multiplication thread]{\label{sfig:mul}%
\begin{tikzpicture}
  \setlength{\rh}{8cm}
  \setlength{\rw}{7cm}
  \setlength{\lh}{0.3cm}
  \setlength{\cw}{0.5cm}
  \draw [black,fill=white] (0,0) rectangle (\rw,\rh);
  \draw  (2\cw,0) -- (2\cw,\rh-4\lh);
  \draw[dashed] (2\cw,\rh-4\lh) -- (2\cw,\rh);

  \draw (0,\lh) -- (\rw,\lh);

  \draw [black] (2\cw+01*0.5\cw, 0.5\lh) node {\tiny $\#$};
  \draw [black] (2\cw+02*0.5\cw, 0.5\lh) node {\tiny $0$};
  \draw [black] (2\cw+03*0.5\cw, 0.5\lh) node {\tiny $1$};
  \draw (3.25\cw, 0.05\lh) rectangle (3.75\cw, 0.95\lh);
  \draw [black] (2\cw+04*0.5\cw, 0.5\lh) node {\tiny $0$};
  \draw [black] (2\cw+05*0.5\cw, 0.5\lh) node {\tiny $\#$};
  \draw [black] (2\cw+06*0.5\cw, 0.5\lh) node {\tiny $0$};
  \draw [black] (2\cw+07*0.5\cw, 0.5\lh) node {\tiny $0$};
  \draw (5.25\cw, 0.05\lh) rectangle (5.75\cw, 0.95\lh);
  \draw [black] (2\cw+08*0.5\cw, 0.5\lh) node {\tiny $0$};
  \draw [black] (2\cw+09*0.5\cw, 0.5\lh) node {\tiny $\#$};
  \draw [black] (2\cw+10.5*0.5\cw, 0.5\lh) node {\tiny $\cdots$};
  \draw [black] (2\cw+12*0.5\cw, 0.5\lh) node {\tiny $\#$};
  \draw [black] (2\cw+13*0.5\cw, 0.5\lh) node {\tiny $1$};
  \draw [black] (2\cw+14*0.5\cw, 0.5\lh) node {\tiny $1$};
  \draw (8.75\cw, 0.05\lh) rectangle (9.25\cw, 0.95\lh);
  \draw [black] (2\cw+15*0.5\cw, 0.5\lh) node {\tiny $0$};
  \draw [black] (2\cw+16*0.5\cw, 0.5\lh) node {\tiny $\#$};
  \draw [black] (2\cw+17.5*0.5\cw, 0.5\lh) node {\tiny $\cdots$};

  \draw [fill=blue!20] (0,0\lh) rectangle (2\cw,1\lh);
  \draw (\cw, 0.5\lh) node {$\times$};
  \draw [fill=red!20] (0,1\lh) rectangle (2\cw,2\lh);
  \draw (\cw, 1.5\lh) node {$z_n$};
  \draw [fill=red!40] (0,2\lh) rectangle (2\cw,3\lh);
  \draw (\cw, 2.5\lh) node {$z_n$};

  \draw (\cw, 14\lh) node { \vdots };

\begin{scope}
  \draw[draw=green, fill=green!40, opacity=0.5] (2.5\cw, 1\lh) -- ++(0,1\lh) parabola
++(4.0\cw, 10.0\lh) -- ++(90:2\lh) parabola ++(-1.75\cw, 4\lh) parabola ++(3\cw,
8\lh) -- ++(0:4\cw) parabola[bend at end] ++(-8\cw, -25\lh) -- ++(0:-1\cw);

  \draw[fill=blue!20] (5.5\cw, 5\lh) rectangle (6\cw, 6\lh);
  \draw[fill=yellow!40] (6\cw, 5\lh) rectangle (7\cw, 6\lh);
  \draw[fill=blue!20] (6.0\cw, 15\lh) rectangle (6.5\cw, 16\lh);

  \draw[->, >=latex] (2.5\cw, 1.5\lh) -- ++(20:7\lh) -- ++(160:3\lh) -- ++(20:1.75\lh) --
++(160:1\lh) -- ++(20:3\lh) -- ++(160:2\lh) -- ++(20:3\lh) -- ++(160:1\lh) --
++(20:2\lh) -- ++(160:2.5\lh) -- ++(20:4\lh) -- ++(160:2\lh) -- ++(20:3\lh) --
++(160:7\lh) -- ++(20:4\lh) -- ++(160:1.5\lh) -- ++(20:6\lh) -- ++(160:1\lh) --
++(20:1.5\lh) -- ++(160:2.5\lh) -- ++(20:3\lh) -- ++(160:4\lh) -- ++(20:5.5\lh);
\end{scope}

\begin{scope}
  \draw[draw=green, fill=green, opacity=0.5] (2.5\cw, 1.5\lh)++(20:7\lh) --
++(0.5\cw,1\lh) parabola ++(6.5\cw, 10.0\lh) -- ++(90:1\lh) parabola ++(-0.5\cw,
3\lh) parabola ++(1.0\cw, 2\lh) -- ++(90:-9\lh) parabola[bend at end] ++(-6.5\cw,
-8\lh) -- ++(0:-1\cw);
  \draw[fill=blue!20] (10\cw, 6\lh) rectangle (10.5\cw, 7\lh);
  \draw (2.5\cw, 1.7\lh)++(20:7\lh) -- ++(20:8\lh) -- ++(160:1\lh)
-- ++(20:3.25\lh) -- ++(160:1.0\lh) -- ++(20:1.8\lh) -- ++(160:1.8\lh) --
++(20:3\lh) -- ++(160:2\lh) -- ++(20:3\lh) -- ++(160:2\lh) -- ++(20:2\lh) --
++(160:1\lh) -- ++(20:1.2\lh);
  \draw[->, >=latex] (14\cw, 17\lh) -- ++(160:1.5\lh) -- ++(20:1.5\lh) -- ++(160:1\lh) -- ++(20:1.0\lh);
\end{scope}

\begin{scope}[xshift=2.9\cw,yshift=9.1\lh]
  \draw[draw=green, fill=green!75!yellow, opacity=0.5] (2.5\cw, 1.5\lh)++(20:7\lh) --
++(0.2\cw,1\lh)
parabola ++(2.25\cw, 3.7\lh)
parabola ++(-0.5\cw, 2.5\lh)
parabola ++(2.7\cw, 4.5\lh) -- ++(0,-3\lh)
parabola[bend at end] ++(-1.4\cw, -2.7\lh)
parabola[bend at end] ++(0.5\cw, -2.7\lh)
parabola[bend at end] ++(-3.5\cw, -3.2\lh);
  \draw[fill=blue!20] (9\cw, 6.75\lh) rectangle (9.5\cw, 7.75\lh);
  \draw[->, >=latex] (2.5\cw, 1.75\lh)++(20:7\lh) -- ++(20:5\lh) -- ++(160:2\lh)
-- ++(20:3.25\lh) -- ++(160:2.0\lh) -- ++(20:1.8\lh) -- ++(160:2.2\lh) --
++(20:3\lh) -- ++(160:2\lh) -- ++(20:3\lh) -- ++(160:1\lh) -- ++(20:1.5\lh);
\end{scope}

\begin{scope}[xshift=5\cw,yshift=20.1\lh]
  \draw[draw=green, fill=green!70!blue, opacity=0.5] (2.5\cw, 1.5\lh)++(20:7\lh) --
++(0.2\cw,1\lh) parabola ++(2.1\cw, 1.0\lh) -- ++(0.2\cw, 0) -- ++(0,-1\lh)
parabola[bend at end] ++(-2.1\cw, -1\lh);
  \draw[fill=blue!20] (6.75\cw, 4.95\lh) rectangle (7.25\cw, 3.95\lh);
  \draw[->, >=latex] (2.5\cw, 1.75\lh)++(20:7\lh) -- ++(20:4.5\lh);
\end{scope}
  \fill (4\cw, 10\lh) -- node[sloped,above] {\parbox{2cm}{\centering\tiny New head created}} ++(0.5\cw, 0);
  \draw[red, ->, >=latex] (4\cw, 13.2\lh) -- (9.5\cw, 13.2\lh);

  \fill (4\cw, 17\lh) -- node[sloped,above] {\parbox{2cm}{\centering\tiny Backtrack on demand}} ++(0.5\cw, 0);
  \draw[red, ->, >=latex] (4\cw, 16.5\lh) -- (11.25\cw, 16.5\lh);

\end{tikzpicture}}
\caption{Computation threads}
\end{figure}

By diagonalization, it is easy to see that there are functions on the real numbers, computable by ITCAs in $\omega$ steps but that are not computable by BSS machines.

\section*{Concluding remarks}

We finish this paper by pointing in a direction that we believe to be promising.

Koepke \cite{Koe05} has defined a transfinite time computation model based on Turing machines that has transfinite space. It can thus compute on arbitrary ordinals and sets of ordinals.
Koepke and Siders \cite{KS05} have also extended the infinite time register machines model to machines with registers containing arbitrary ordinals. These models make it possible to compute the bounded truth predicate of the constructible universe, $\{ (\alpha, \varphi, \vec{x}) : \alpha \in \text{Ord}, \varphi \text{ an $\in$-formula}, \vec{x}\in L_\alpha, L_\alpha \models \varphi(\vec{x}) \}$, and allows the following characterization of computable sets of ordinals: a set of ordinals is ordinal computable from of a finite set of ordinal parameters if and only if it an element of Gödel's constructible universe $L$. 
Ordinal computability is moreover interesting to be able to reprove some facts about $L$.

We propose the following definition for ordinal computations over cellular automata.

\begin{defi}
\label{def:oac}
An {\em ordinal cellular automaton} $A$ is defined by $\Sigma$, the finite set
of states of $A$, linearly ordered by $\prec$ and with a least element
$\mathbf{0}$;
and $\delta : \Sigma^3 \to \Sigma$, the local rule of $A$,
satisfying $\delta(\mathbf{0},\mathbf{0},\mathbf{0})=\mathbf{0}$, so that
$\mathbf{0}$ is a {\em quiescent} state.

A {\em configuration} of length $\xi$ is an element of $\Sigma^{\xi}$. The
local rule $\delta$ induces on configurations of length $\xi$ a global rule
$\Delta: \Sigma^{\xi} \to \Sigma^{\xi}$ such that

$$\left\{\begin{array}{ll}
\Delta(C)_0 = \delta (\liminf^{\prec}_{\gamma<\xi} C_{\gamma},C_0,C_1), \\
\Delta(C)_{\beta+1} = \delta (C_{\beta},C_{\beta+1},C_{\beta+2}), \\
\Delta(C)_{\lambda} = \delta (\liminf^{\prec}_{\gamma<\lambda} C_{\gamma},
C_{\lambda}, C_{\lambda+1}) & \text{if $\lambda$ limit and $\lambda > 0$.}
\end{array}\right.
$$

Starting from a configuration $C\in\Sigma^{\xi}$, the {\em evolution} of
length $\theta$ of $A$ is given by $(\Delta^\alpha(C))_{\alpha \leqslant
\theta}$:
$$
\begin{array}{lcl}
\Delta^{\beta+1}(C) & = &\Delta (\Delta^{\beta} (C)) \\
\Delta^{\lambda}(C)_{\iota} & =  &\liminf^{\prec}_{\gamma<\lambda} \Delta^{\gamma}(C)_{\iota} \text{ for all $\iota < \xi$ and $\lambda$ limit}
\end{array}
$$
\end{defi}

We think that it would be interesting to try to carry the results of the other ordinal machines models over to this ordinal cellular automata model. In this model, there is the limitation due to the fixed length of configurations, which is imposed by the cylindrical nature of our model. There are certainly other ways to allow information to flow from the right to the left of the configurations (and not only left to right).

%

\end{document}